\documentclass[10pt]{amsart}
\usepackage[utf8]{inputenc}
\usepackage{amsmath,amssymb,amsthm}
\usepackage{enumitem}
\usepackage[hidelinks]{hyperref}
\usepackage{cleveref}
\usepackage{tikz}
\usepackage{enumitem}
\usetikzlibrary{patterns}

\DeclareMathOperator{\SL}{\mathrm{SL}}

\newcommand{\N}{\mathbb{N}}
\newcommand{\Z}{\mathbb{Z}}
\newcommand{\R}{\mathbb{R}}
\newcommand{\C}{\mathbb{C}}

\renewcommand{\P}{\mathbb{P}}
\newcommand{\E}{\mathbb{E}}

\newcommand{\vp}{\varphi}
\newcommand{\ve}{\varepsilon}

\newtheorem{thm}{Theorem}[section]

\newtheorem{lem}[thm]{Lemma}

\newtheorem{prop}[thm]{Proposition}

\theoremstyle{definition}

\newtheorem{deffo}[thm]{Definition}
\newtheorem{remm}[thm]{Remark}

\numberwithin{equation}{section}

\title[Heavy tailed Anderson models]{Localization for heavy tailed Anderson models}

\author{Omar Hurtado}

\begin{document}
	
\begin{abstract}
    Using recent results on uniform large deviation estimates for random matrix products obtained in \cite{hurtado2025uniformestimatesrandommatrix}, we prove localization for one dimensional Anderson models with heavy tails.
\end{abstract}
\maketitle

\section{Introduction}
\subsection{Main results}
This article is concerned with the Anderson model with heavy tailed distributions. This is a fundamental model in the study of disordered materials, whose spectral properties have been extensively studied. Mathematically it is given by a (random) operator acting on $\ell^2(\Z)$ of the form

\begin{equation}\label{anderson} [H\psi](n) = \psi(n+1) + \psi(n-1) + V_n \psi(n) \end{equation}
where $V_n$ is an i.i.d. random potential. Such random operators allow one to study the effect on electronic transport of disorder in a material, and said operators are expected (and to a large degree, now known) to display a phenomenon known as Anderson localization. For any probability measure on $\mu$, we call the operator $H$ defined by \Cref{anderson} with $V_i$ having law $\mu$ the Anderson model with single site distribution $\mu$. There are stronger and weaker notions of localization, the stronger notions formulated in terms of moments of the position operator applied to the time evolution associated to $H$; this article focuses on purely spectral notions of localization.

\begin{deffo}
    An operator $H$ is spectrally localized if it has no continuous spectrum and its eigenvectors decay exponentially.
\end{deffo}

The current state of the art concerning spectral localization for general operators of the form (\ref{anderson}) is the following, from \cite{carmona1987anderson} by  Carmona, Klein and Martinelli:
\begin{thm}[\cite{carmona1987anderson}]
    If $\mu$ satisfies
    \begin{equation}\label{fractionalmoment}
        \int |x|^\alpha \,d\mu(x) < \infty
    \end{equation}
    for some $\alpha > 0$ and is not concentrated on a single point, with single site distribution $\mu$ is almost surely spectrally localized.
\end{thm}

Importantly, Carmona, Klein and Martinelli make no regularity assumptions in their work, and we will not either. In this work we treat models which satisfy the weaker moment condition:

\begin{equation}\label{logmoment}
    \int (\log^+(|x|))^p\,d\mu(x) < \infty
\end{equation}
(Here $\log^+(x) = \max\{\log x, 0\}$.) We obtain certain non-trivial estimates as soon as $p > 1$, but our main result concerns the case where $p > 11$.

By applying results obtained in \cite{hurtado2025uniformestimatesrandommatrix} and combining them with variations on arguments from \cite{carmona1987anderson} and a result from \cite{von1989new} which we use as a black box, we obtain the following theorem:

\begin{thm}\label{mainlocalthm}
    If $\mu$ satisfies \Cref{logmoment} for $p > 11$ and is not supported on a single point, then the Anderson model associated to $\mu$ is almost surely spectrally localized.
\end{thm}

At least in terms of relaxing the moment assumptions on $\mu$, this is the first improvement on the work of Carmona, Klein and Martinelli which does not require any regularity assumption, and partially answers a conjecture posed in \cite{Macera2022}.

\subsection{Background on random Schr\"odinger type operators}

These models have been quite well studied, and are fundamental models in the study of disordered materials going back to the groundbreaking (and Nobel Prize winning) work of P. W. Anderson near the middle of the twentieth century. We do not try and give a full account of the history, but we mention some useful reference works, as well as works concerning localization for similar models, with special attention paid to those which treat singular distributions in one dimension.

Indeed, good accounts for the general, possibly multidimensional, case appear in \cite{cycon2009schrodinger,aizenman2015random}. The one dimensional theory is markedly different because of the transfer matrix method. While there are other methods which are unique to one dimension, e.g. the spectral averaging of \cite{kunz1980spectre}, the modern study of random one-dimensional models (and closely related one-dimensional ergodic models of various other flavors) is dominated by the use of this method, which studies ergodic operators by studying associated linear cocycles. For an account of the specifically one dimensional theory, discussing also work for more general ergodic models, see e.g. \cite{damanik2025one}.

While there were many results capable of treating the case where the random noise is sufficiently regular (see \cite{frohlich1983absence,frohlich1985constructive,kunz1980spectre,goldsheid1977random}), the first result capable of treating the case with singular potentials (e.g. $V_n$ Bernoulli variables) was the groundbreaking work of Carmona, Klein, Martinelli in \cite{carmona1987anderson}. Already, some aspects of the theory of random matrix products were exploited here to facilitate the proof. This proof combined the Multiscale Analysis (MSA) method developed in \cite{frohlich1983absence} with certain facts from the theory of random matrix products. 

Since then, various ``single-scale'' proofs of localization in one dimension which eschew MSA and leverage even further properties of random matrix products have appeared, see e.g. \cite{Jitomirskaya2019,Bucaj2017LocalizationFT,gorodetski2021parametric}. See also \cite{rangamani2019singular, Macera2022, hurtado2023lifting} for more works which used methods introduced in \cite{Jitomirskaya2019} to prove localization for wider classes of random operators, and \cite{gorodetski2025non,damanik2025localization} which built on methods from \cite{gorodetski2021parametric}. (It is worth mentioning also the important paper \cite{shubin1998some} which provided a proof based on harmonic analysis via results in harmonic analysis concerning a variant of the uncertainty principle.) In the ``single-scale'' proofs, uniform large deviation estimates are of crucial importance. It would be interesting to see if a ``single-scale'' proof of the present results is possible; all of the current single-scale approaches leverage exponential large deviation estimates, which are not available in the heavy-tail regime.

As this history makes clear, lack of regularity in $\mu$ complicates proofs of localization. (This is also true in higher dimensions, and in fact localization for e.g. the Bernoulli-Anderson model was only recently solved in dimensions 2 \cite{ding2020localization} and 3 \cite{li2022anderson}, and remains open in dimensions 4 and higher; this work used ideas developed for related continuum models in the landmark work of Bourgain and Kenig \cite{bourgain2005localization}.) The highly technical version of MSA developed in \cite{bourgain2005localization, klein2012comprehensive} and various works which built upon these excepted, applications of MSA generally require some moment condition but can treat H\"older regular potentials in higher dimensions. As soon as one has e.g. $\mu$ given by a bounded density with respect to Lebesgue measure, one does not require any moment condition. Indeed, the celebrated fractional moment method introduced by Aizenman and Molchanov in \cite{aizenman1993localization} gives the following:

\begin{thm} \cite{aizenman1993localization}
    Let $\mu$ be a distribution generated by a bounded density on $\R$. Then the Anderson model generated by $\mu$ is almost surely spectrally localized.
\end{thm}

In particular, heavy tailed models with regular potentials can be treated via FMM; our localization results are novel in the context of potentials with heavy tails and a lack of regularity.
\subsection{Approach of the paper}
The main new tool is the use of certain uniform large deviation estimates which were proven by Raman and the author in \cite{hurtado2025uniformestimatesrandommatrix}. The results of said paper are reasonably general, requiring only irreducibility and proximality assumptions. Specialized to the context of the current paper, we specifically obtain polynomial large deviation estimates for the associated transfer matrices when \Cref{logmoment} is satisfied for $p>1$; this is \Cref{ldesthm}. This more or less automatically yields one of the necessary assumptions for a localization proof via MSA for $p$ sufficiently large; the so-called initial scale estimate \Cref{initialscalethm2}.

These polynomial large deviation estimates also allow us to prove the second necessary assumption; the Wegner estimate which controls the probability of an eigenvalue falling with a small interval, in this case \Cref{wegthm}. However, here it is not quite automatic, and we prove it by adapting an argument appearing in \cite{carmona1987anderson} to prove a similar bound, though it is worth pointing out that we obtain a weaker bound under weaker assumptions. We prove a polynomial bound, and though the outline is the same as in \cite{carmona1987anderson} there are a few technical difficulties only arising in our case, and at the same time we simplify certain parts of the argument using our bounds which are uniform in energy. (Such uniform bounds were not yet proven when \cite{carmona1987anderson} was published, though the pointwise versions were and figured crucially in the proof.) The polynomial bounds we obtain suffice to prove localization via a theorem of von Dreifus and Klein in \cite{von1989new}, which makes use of MSA; we use it as a black box and do not carry out an MSA argument ourselves.
\section*{Acknowledgements}
The author thanks Lana Jitomirskaya for suggesting the question of localization for heavy tailed distributions. The author was supported in part by NSF Grants Nos.
DMS-2052899 and DMS-2155211, and Simons Grant No. 896624.
\section{Preliminaries}
\subsection{Notation}

Throughout, we always use $\P$ to denote the probability of an event. Generally, we use calligraphic letters to denote events, and $\mathcal{A}^C$ is the complement of the event $\mathcal{A}$. For an operator $H$ on a Hilbert space, $\sigma(H)$ denotes the spectrum of $H$, i.e. the set of those $E \in \C$ such that $H-E$ does not have a bounded inverse. All operators considered are self-adjoint (by our assumption that $V$ is real-valued) and so in particular the spectrum is actually a subset of $\R$. Whenever we deal with random phenomena we use $\P[\mathcal{A}]$ to denote probability of an event $\mathcal{A}$; for a random variable, we denote the probability of an event depending on its value by e.g. $\P[X \leq \ve]$.

\subsection{Generalized eigenfunctions and transfer matrices}

Much of this section recalls facts which are well known in the study of random or more generally ergodic Schr\"odinger operators; the fundamentally new inputs are \Cref{ldesthm} (which is proven in \cite{hurtado2025uniformestimatesrandommatrix}) and \Cref{initialscalethm}. More comprehensive accounts of many of these identities appear in e.g. \cite{Jitomirskaya2019, rangamani2019singular, hurtado2023lifting}.

Fundamental objects in the study of spectral questions concerning the Anderson model are generalized eigenfunctions:

\begin{deffo}
    We say $E$ and $\psi \in \R^\Z$ respectively are a generalized eigenvalue and generalized eigenfunction respectively if
    \begin{equation}\label{eigeneq}
        H\psi = E\psi
    \end{equation}
    and $\psi$ is of strictly slower than exponential growth, i.e.
    \[ \limsup_{|n| \rightarrow \infty} \frac{1}{|n|} \log |\psi(n)|\leq 0 \]
\end{deffo}

Clearly all bona fide $\ell^2$ eigenfunctions are generalized eigenfunctions, but formal solutions with e.g. very slow decay or even polynomial growth are also permitted. It is a general fact that ``intermediate'' growth which is superpolynomial but subexponential is impossible for formal solutions, so one can restrict to the class of polynomially bounded formal solutions, but this fact is not important for our purposes.

The importance of these objects in spectral theory is a consequence of a result known as Sch'nol's theorem, which roughly says that the spectrum of $H$ is supported on the set of the generalized eigenvalues of $H$. We will not make this precise since we do not need to make use of Sch'nol's theorem directly, but it is implicit via our use of \Cref{msaassumptions} below (which is a special case of a result of von Dreifus and Klein). Because of this relationship, in localization arguments asymptotics of formal solutions end up being central. In the one dimensional context, these can be studied fruitfully via the transfer matrix method. Specifically, we introduce first the one step transfer matrices:

\begin{equation}
    T_n^E = \begin{pmatrix} E - V_n & -1 \\ 1 & 0\end{pmatrix}
\end{equation}

For $\psi$ solving $H\psi = E\psi$, we have:

\begin{equation}\label{onesteptmeq}
    \begin{pmatrix} \psi(n+1) \\ \psi(n) \end{pmatrix} = T_n^E \begin{pmatrix} \psi(n) \\ \psi(n-1) \end{pmatrix}
\end{equation}

Naturally, one can multiply one step transfer matrices to get a matrix which gives one $\begin{pmatrix} \psi(n+1) \\ \psi(n) \end{pmatrix}$ from any $\begin{pmatrix} \psi(m+1) \\ \psi(m) \end{pmatrix}$; we define:

\[
    T_{n,m}^E := \begin{cases}
        T_n^ET_{n-1}^E\cdots T^E_{m+1}, \quad &n>m\\
        I, \quad &n=m\\
        (T_{n+1}^E)^{-1}(T_n^E)^{-1}\cdots (T_m^E)^{-1}, \quad &n<m
    \end{cases}
\]

By induction, \Cref{onesteptmeq} gives the following:
\begin{equation}\label{tmeq}
    \begin{pmatrix} \psi(n+1)\\ \psi(n) \end{pmatrix} = T_{n,m}^E \begin{pmatrix} \psi(m+1)\\ \psi(m) \end{pmatrix}
\end{equation}

Note that $T_{n,m}^E$ is a product of $|n-m|$ i.i.d. random matrices, with distribution of each term depending on the sign of $n-m$; we will focus on the case $n>m$ specifically. We let $\mu^E$ be the distribution of $T_0^E$ on $\SL(2,\R)$. Under certain assumptions on the distribution, some of them ``geometric'' and some in terms of moments, we can obtain exponential behavior of random products with high probability. By exponential behavior, we refer to linear growth of $\log\|T_{[a,b]}\|$ and related quantities in $|b-a|$; in particular the following well-known fact is a consequence of work of Furstenberg and collaborators:
\begin{prop}[\cite{FurstenbergNRP,FurstenbergKifer,FurstenbergKesten}]
    Under the assumptions of \Cref{mainlocalthm}, there is a function $\lambda:\R \rightarrow \R$ which is strictly positive and continuous such that
    \[\lim_{b-a \rightarrow +\infty} \frac{1}{b-a} \log\|T_{[a,b]}^E\| = \lambda(E)\]
\end{prop}
A similar statement holds for the quantities $\log \|T_{[a,b]}^E x\|$ and $\log|\langle x, T_{[a,b]}^E y \rangle|$, where $x, y$ are non-zero elements of $\R^2$. (There is also $\lambda^-(E)$ corresponding to the limit $b-a \rightarrow - \infty$; in many cases of interest, including the particular case of Schr\"odinger cocycles it turns out that this ``backwards'' Lyapunov exponent coincides with the ``forwards'' one.)
\begin{remm}
    To be precise all the works cited treat products of random matrices, and make various assumptions, said assumptions varying from paper to paper. Furstenberg and Kesten demonstrated the existence of the limit under very weak assumptions; later Furstenberg showed positivity under assumptions which hold in the present context. Continuity was demonstrated by Furstenberg and Kifer under assumptions which also hold in the present context. That the various hypotheses necessary to apply these results are satisfied is a standard argument, see e.g. \cite{carmona1987anderson}.
\end{remm}
Under strong moment conditions (like e.g. those assumed in the paper of Carmona, Klein and Martinelli) strong quantitative bounds on this convergence are known; these were proven (in a pointwise form) by Le Page in \cite{les1982theoremes} and exploited in the work of Carmona, Klein and Martinelli. Versions uniform in energy have since been proven in e.g. \cite{tsay1999some,Bucaj2017LocalizationFT,duarte2016lyapunov, duarte2020large}; note that except for the work of Tsay all of these results assume boundedness, which is stronger than the assumptions made by Carmona, Klein and Martinelli. 

In \cite{hurtado2025uniformestimatesrandommatrix}, Raman and the author showed (among other things) uniform estimates under weaker moment assumptions than those previously considered, in \cite[Theorem 1.20]{hurtado2025uniformestimatesrandommatrix}.

\begin{thm}[\cite{hurtado2025uniformestimatesrandommatrix}]\label{ldesthm}
We let $H$ be a random operator as in \Cref{mainlocalthm}, and $T_{[a,b]}^E$ the associated transfer matrices. Then for any $I \subset \R$ compact and $\ve > 0$, there is $C> 0$ (depending on $\mu$, $I$, $\ve$) such that all of the following estimates hold:

\begin{align*}\P[|\log\|T_{[a,b]}^E\| - \lambda(E)|>n\ve] &\leq CL^{-p}& \\
    \P[|\log\|T_{[a,b]}^E x\| - n\lambda(E)| > n\ve ] &\leq CL^{-p}\quad& \text{for}\quad &\|x\|=1\\
    \P[|\log|\langle x, T_{[a,b]}^E y\rangle - n\lambda(E)| > n\ve ] &\leq CL^{-p}\quad& \text{for}\quad &\|x\|=\|y\|=1
\end{align*}
\end{thm}

\begin{remm}
    In fact, these bounds hold for operators of the form \Cref{anderson} satisfying \Cref{logmoment} with $p \geq 3$. Even for those satisfying it for some $p \in (1,3)$ a similar bound holds but with $CL^{-p}$ replaced by $CL^{1.5 - 1.5p}$.
\end{remm}

These bounds on the transfer matrices are the main new input enabling the localization result. They enable a Wegner estimate, though the proof is somewhat technical and deferred to the last section. They also provide nice bounds on what are called the Green's functions of the operator, which is a sort of finite volume resolvent.

In order to introduce the Green's function, we must first define a few other quantities. First we let $H_{[a,b]}$ denote the truncation of $H$ to the interval $[a,b]$. (More precisely, $H_{[a,b]}$ is the corner $P_{[a,b]} H P_{[a,b]}$, where $P_{[a,b]}$ is the appropriate projection. For any integers $a< b$ and $x,y \in [a,b]$ (here and throughout $[a,b]$ denotes the interval in $\Z$), we define the Green's function as follows:

\begin{equation}\label{greensdef}
    G^E_{[a,b]}(x,y) = \langle \delta_x, (H_{[a,b]}-E)^{-1} \delta_y \rangle
\end{equation}
and formally set $G^E _{[a,b]}= \infty$ if $E \in \sigma(H_{[a,b]})$. (That is, we will use e.g $\P[ G^E_{[a,b]} > C]$ as a shorthand for $\P[ E \notin \sigma(H_{[a,b]})\quad \text{and}\quad G^E_{[a,b]} > C]$.)
\begin{remm}
    $H_{[a,b]}$ can be quite naturally identified with an operator on $\ell^2([a,b])$, i.e. a $(b-a+1)^2$ matrix in the standard basis. In particular $G^E_{[a,b]}$ is defined in terms of this identification; by $(H_{[a,b]}-E)^{-1}$ we mean the $\ell^2([a,b])$ inverse, not the inverse with respect to the whole space.
\end{remm}

The Green's function is a key object because it is closely tied to the behavior of formal solutions to \Cref{eigeneq}, but it satisfies certain useful identities which make it more amenable to analysis than the formal solutions themselves. Specifically, we have the following well known identity relating the Green's function and any formal solution of \Cref{eigeneq}:
\begin{equation}\label{poisson}
    \psi(x) = -G^E_{[a,b]}(x,a)\psi(a-1) - G^E_{[a,b]}(x,b)\psi(b+1) \quad\text{for}\quad x \leq y
\end{equation}

At the same time, the Green's function satisfies many well-known identities which allow one to relate its asymptotics to those of the transfer matrices, as one would suspect from the relation to the formal solutions via \Cref{poisson}.

We let
\[ P_{[a,b]}^E = \begin{cases}
    1\quad &\text{for}\quad b\leq  a\\
    \det T_{[a,b]}^E\quad&\text{for}\quad a < b
\end{cases}\]
Then we have the following:
\[|G_{[a,b]}^E(x,y)| = \frac{|P_{[a,x-1]}^E|\cdot |P_{[y+1,b]}^E|}{|P_{[a,b]}^E|}\quad \text{for}\quad a\leq x \leq y \leq b\]
In particular,
\begin{equation}\label{greens2det}\begin{split}|G_{[x-L/2,x+L/2]}^E(x,x+L/2)| &= \frac{|P_{[x-L/2,x+1]}^E|}{|P_{[x-L/2,x+L/2]}^E|}\\
|G^E_{[x-L/2,x+L/2]}(x,x-L/2)| &= \frac{|P^E_{[x-1,x+L/2]}|}{|P^E_{[x-L/2,x+L/2]}|}
\end{split}
\end{equation}
(Note that we use self-adjointness of $(H_{[a,b]}-E)^{-1}$ for the second calculation.)

Finally, we recall the following well-known fact:

\begin{equation}\label{det2transfer} P_{[a,b]}^E = \left\langle \begin{pmatrix} 1 \\ 0 \end{pmatrix}, T_{[a,b]}^E \begin{pmatrix} 1 \\ 0 \end{pmatrix} \right\rangle \quad \text{for}\quad a \leq b\end{equation}

From these identities we can extract one of the crucial ingredients towards a localization result: probabilistic estimates on the Green's function. While the large deviation results quoted above are a new input, the basic strategy goes back to work of Fr\"ohlich and Spencer.

\begin{prop}\label{initialscalethm}
    Fix $I \subset \R$ a compact interval, and $H$ an Anderson model with single site distribution $\mu$ not supported on a single point and satisfying \Cref{logmoment} for some $p > 3$. Then there is $C > 0$ (depending on $I$) such that

    \[\P[ |G_{[x-L/2,x+L/2]}(x,x\pm L/2)| \geq e^{\lambda(E)L/4}] \leq CL^{-p}\]
    for all $L \in \N$ and $E \in I$.
\end{prop}

Before commencing the proof in earnest, we define
\[ \lambda_{\min} = \min_{E \in I} \lambda(E)\]
We will use this shorthand throughout.

\begin{proof}
    By \Cref{ldesthm}, we have, for any $\ve > 0$ and compact $I\subset \R$, the existence of $C,c >  0$ (depending on $I$ and $\mu$) such that with probability at least $1 - L^{1-p''}$, we have the following bounds:
    \begin{align*}
    |P_{[x-L/2,x+L/2]}| &\geq e^{(L+1)(\lambda(E)-\ve)}\\
    |P_{[x-L/2,x+1]}| &\leq e^{(L/2+1)(\lambda(E)-\ve)}\\
    |P_{[x-1,x+L/2]}| &\leq e^{(L/2+1)(\lambda(E)-\ve)}
    \end{align*}
    for any $p'' \in (p', p)$ and $L$ sufficiently large. 
    Indeed, \Cref{det2transfer} gives precisely this. One thus obtains, via \Cref{greens2det}, that with probability at least $1-L^{1-p'}$, we have
    \[ |G_{[x-L/2,x+L/2]}^E(x,x\pm L/2)| \leq e^{(3L/2 + 2)\ve-\lambda(E)L/2 }\]
    Taking $\ve < \lambda_{\min}/6$ gives the result.
\end{proof}
    
\section{Multiscale analysis}
As has been discussed in the introduction, multiscale analysis (MSA) is a central technique in the study of random Schr\"odinger operators. In some sense, the strongest version of MSA known to work with unbounded operators is the bootstrap MSA developed by Germinet and Klein in \cite{germinet2001bootstrap}; stronger versions have since appeared (see e.g. \cite{bourgain2005localization,klein2012comprehensive,ding2020localization,rangamani2023dynamical}), but to our knowledge the arguments bridging from the very weak Green's functions estimates this new MSA can obtain to localization have only been done in the bounded context. 

In general, a multiscale analysis is an inductive method by which one can establish bounds on what is called the Green's function. 
While we have already defined one type of ``interval'', in this section we will also consider another type of interval in order to align with usual MSA arguments. For $L$ even, we let $\Lambda_L(x)$ denote the interval $[x-L/2,x+L/2]$; note that unless we are discussing multiple possibly overlapping boxes, the choice of center does not matter for any probabilistic questions; this is a consequence of stationarity and independence of the potentials at different sites; in particular our theorem statements will discuss $\Lambda_L := \Lambda_L(0)$ explicitly, but it should be understood that such results hold for any choice of center. MSA requires the notion of regular ``boxes'' (intervals in the 1D setting). There are various notions appearing across different papers, but we use (a slight variation of) the notion appearing in \cite{von1989new}.

\begin{deffo}
    We say $\Lambda_L(x)$ is $(m,E)$ regular if
    \[ |G_{[x-L/2,x+L/2]}(x,x\pm L/2)| \leq e^{-mL/2}\]
\end{deffo}

MSA is a method whereby we can obtain bounds on $\P[\Lambda_L\text{ is } (m_0,E)\text{-regular}]$; the essential insight is that if most boxes at a smaller scale are regular, and we can avoid a ``resonance'', then we will have regularity (with a slight loss to exponential rate $m$) at larger scales.

In particular, von Dreifus and Klein proved the following in \cite{von1989new} specifically in the context of Schr\"odinger operators on $\Z$:

\begin{thm}[\cite{von1989new}]\label{msaassumptions}
    Let $H$ be an Anderson model of the form (\ref{anderson}). Let $E_0 \in \R$. Suppose that there is $L_0 >0$ such that
    \begin{enumerate}[label=(P\arabic*)]   \item\label{initialscale}$\P[\Lambda_{L_0}\text{ is } (m_0,E_0)\text{-regular}] \geq 1- L_0^{-q_1}$ for some $q_1 > 1$
        \item\label{Wegner}There exist $\beta \in (0,1)$, $q_2 > 4q_1 + 6$ and $\ve> 0$ such that $\P[ \mathrm{dist}(\sigma(H_{\Lambda_L}), E) \leq e^{-L^\beta}] \leq 1/L^{q_2}$ for some all $L \geq L_0$ and $E$ satisfying $|E-E_0| \leq \ve$
    \end{enumerate}
     Then given any $m \in (0,m_0)$, there exists $B < \infty$ (depending on the various parameters at play excluding $L_0$) such that if $L_0 > B$, we can find $\delta$ (depending on the various parameters at play) such that with probability one, $H$ is exponentially localized in $(E_0 - \delta, E_0 + \delta)$, i.e. the portion of the spectrum contained in $(E_0 - \delta, E_0+\delta)$ is pure point and all eigenfunctions associated to $E$ in this range decay exponentially with rate $m$.
\end{thm}
\begin{remm}We have formulated a specifically one dimensional version of the result from \cite{von1989new}; the result is true in arbitrary dimension with some modifications. Moreover, there have been significant advances since said work. Notably  a bootstrap MSA was introduced by Germinet and Klein in \cite{germinet2001bootstrap}, which can work with a much weaker estimate than \ref{initialscale}. However, said variant also requires a much stronger version of \ref{Wegner} than we can prove. As has been mentioned before, a variant of MSA capable of dealing with even weaker estimates was introduced and developed in \cite{bourgain2005localization,klein2012comprehensive}, but currently it seems this variant requires boundedness.

\end{remm}

While the theorem of von Dreifus and Klein, on the face of it, only proves localization in small intervals, we will use basic union bounds to extend the result to the whole real line. In particular, the proof of \Cref{mainlocalthm} is more or less reduced to verifying the two conditions condition \ref{initialscale} and condition \ref{Wegner} for all energies. This is the content of \Cref{initialscalethm2} and \Cref{wegthm} below.

The first result, verifying condition \ref{initialscale}, we have more or less already proven; it follows from \Cref{initialscalethm}, and the derivation is very brief.  The second condition, often called a Wegner estimate, we prove using roughly the argument of Carmona, Klein and Martinelli in \cite{carmona1987anderson} to obtain the same under their stronger moment assumption. There are some technical differences in the argument owing to the estimates being uniform in energy but weaker than those exploited in \cite{carmona1987anderson}.

\begin{thm}\label{initialscalethm2}
    Given any compact interval $I \subset \R$, and $H$ with single site distribution $\mu$ satisfying \Cref{logmoment} for $p > 2$, there exists $m > 0$ such that
    \begin{equation}
        \P[\Lambda_L\text{ is } (E,m)\text{-regular}] \leq 1-L^{2-p}
    \end{equation}
    for $L$ sufficiently large.
\end{thm}

\begin{proof}
    For \[m < \frac{1}{4}\lambda_{\min}\] the desired regularity follows by \Cref{initialscalethm}.
\end{proof}

\section{Wegner estimate}
Here we prove the Wegner estimate, verifying that condition \ref{Wegner} holds.
\begin{thm}\label{wegthm}
Given $H$ with single site distribution $\mu$ non-trivial and satisfying \Cref{logmoment} for $p > 2$, and fixed compact $I \subset \R$ we have, for any $\beta \in (0,1)$:
\[ \P[\mathrm{dist}(\sigma(H_{\Lambda_L}),E) \leq e^{-L^\beta}] \leq L^{1-\eta p'}\]
for any $0 < \eta < \beta$ and $p' \in (1,p)$ and $L$ sufficiently large, with the requisite largeness of $L$ depending on the interval $I$, $\beta$, $p'$, $\eta$ and $\mu$.
\end{thm}

Our proof is based off the strategy of \cite[Theorem 4.1]{carmona1987anderson}, and the following lemma is a technical lemma appearing in their paper.

\begin{lem}[\cite{carmona1987anderson}]\label{holdercor}
    Let $H$ be as in \Cref{wegthm}. Then given any $I \subset \R$ compact, there are $\rho$ and $C$ positive, depending on $\mu$ and $I$, such that
    \[ \P[\mathcal{E}_{E,L,\ve}] \leq CL\ve^\rho \]
    where $\mathcal{E}_{E,L,\ve}$ denotes the event that there is an eigenvalue $E'$ of $H_{[-L,L]}$ with a normalized eigenfunction $\vp'$ such that $|\vp'(-L)|^2 + |\vp'(L)|^2 \leq \ve^2$.
\end{lem}
It was proven in \cite{carmona1987anderson} for operators generated by potentials with fractional moments, but the proof only required H\"older continuity of the density of states, which in particular was shown (in the same paper) for non-trivial $\mu$ satisfying \Cref{logmoment} with $p\geq 1$.

The last necessary notion before we prove \Cref{wegthm} a certain function which agrees with a power of the logarithmic function at large values, but is linear at small values, introduced (to our knowledge at least) in \cite{hurtado2025uniformestimatesrandommatrix}. Specifically, we define a function for $x \geq 0$:
\begin{equation}
    \log^{p\star}(x) = \begin{cases}
        (\log x)^p \quad &\text{for}\quad x \geq e^p\\
        \left(\frac{p}{e}\right)^p x \quad &\text{for} \quad x < e^p
    \end{cases}
\end{equation}
These functions have the following important properties for $p\geq 1$:
\begin{enumerate}
    \item $\log^{p\star}$ is concave
    \item $\log^{p\star}(0)$ is subadditive
    \item $\log^{p\star}(xy) \leq C_p \log^{p\star}(x)\log^{p\star}(y)$ for some $C_p \geq 1$
\end{enumerate}
In fact, if $p \geq e$, one can take $C_p = 1$; in general one can achieve $C_p= 1$ by rescaling. Since we are principally interested in the case $p > 11$, we will proceed under the assumption $C_p = 1$ without rescaling and omit it to keep notation lighter.

Crucially, if a measure $\mu$ satisfies \Cref{logmoment} for some $p \geq 1$, then for the same $p$ it satisfies:
\begin{equation}
    \int \log^{p\star}(x)\,d\mu(x)
\end{equation}
However, moments with respect to $\log^{p\star}$ are more tractable due to its various nice properties.

Now we prove \Cref{wegthm}, adapting the argument from \cite{carmona1987anderson}.

\begin{proof}[Proof of \Cref{wegthm}]
    As we have discussed, if $\mu$ satisfies \Cref{logmoment}, then 
    \[ \E[\log^{p\star} (\|T_0^E\|)] < \infty\]
    for any given $E$, and moreover for any compact $I \subset \R$,
    \[ M:= \sup_{E \in I}\E[\log^{p\star}(\|T_0^E\|)]<\infty \]
   Throughout, to keep computations legible, we omit floor symbols; when $L^{\eta}$ appears somewhere an integer argument is necessary we mean $\lfloor L^{\eta} \rfloor$. The $O(1)$ difference is negligible in all calculations and estimates. We consider two families of events:
    \begin{align}
        \mathcal{A}_{\theta,L}^E &:= \left\{\|A_{[-L, -L+L^{\eta}]}^E \begin{pmatrix} 0 \\ 1 \end{pmatrix}\| \geq \exp(\theta L^{\eta})\right\}\\
        \mathcal{B}_{\theta,L}^E &:= \left\{\|A_{[L-L^{\eta}, L]}^E \begin{pmatrix} 1 \\ 0 \end{pmatrix}\| \geq \exp(\theta L^{\eta})\right\}
    \end{align}
Letting $E'$ range over eigenvalues of $H_{[-L,L]}$, we can bound $\P[\mathrm{dist}(\sigma(H_{[-L,L]},E) < e^{-L^\beta}]$ by the sum of probabilities for the following four events, for some $\kappa >0$ to be determined:
\begin{align*}
    \mathcal{C}_1 &:= \{\mathrm{dist}(\sigma(H_{[-L,L]},E) < e^{-L^\beta}]\} \cap \bigcap_{|E'-E|< e^{L^\beta}} (\mathcal{A}_{\kappa/2, L}^{E'} \cap \mathcal{B}_{\kappa/2, L}^{E'})\\
    \mathcal{C}_2 &:= \mathcal{A}_{\kappa,L}^E \cap \mathcal{B}_{\kappa, L}^E \cap \bigcup_{|E'-E| < e^{-L^\beta}} (\mathcal{A}_{\kappa/2, L}^{E'})^C\\
    \mathcal{C}_3 &:= \mathcal{A}_{\kappa,L}^E \cap \mathcal{B}_{\kappa, L}^E \cap \bigcup_{|E'-E| < e^{-L^\beta}} (\mathcal{B}_{\kappa/2, L}^{E'})^C\\
    \mathcal{C}_4 &:=(\mathcal{A}_{\kappa, L}^E)^C \cup (\mathcal{B}_{\kappa, L}^E)^C
\end{align*}

We estimate $\P[\mathcal{C}_1]$ using \Cref{holdercor}. By a standard argument using the Poisson formula (see the proof of \cite[Theorem 4.1]{carmona1987anderson}) the event $\mathcal{C}_1$ implies, for every $E'$ satisfying $|E'-E| < e^{- L^\beta}$, the existence of a corresponding normalized eigenfunction $\vp' \in \ell^2([-L,L])$ with 
\[ \max\{|\vp'(L)|^2, |\vp'(-L)|^2\} \leq 2\exp(\kappa L^{\eta}/2)\]
Note that for large $L$ the intervals $[-L+L^{\eta}, L]$ and $[-L, L-L^{\eta}]$ are of size at least $L/2$; we obtain
\begin{equation}\label{c1bound}\P[\mathcal{C}_1] \leq CL\exp(-\rho\kappa L^{\eta}/2)\end{equation}
by \Cref{holdercor}. (Note that in particular we take $\ve = \exp(-\rho\min\{ \kappa L^{\eta}/2, L^\beta\})$, which for sufficiently large $L$, is just $e^{-\rho \kappa L^{\eta}/2}$).

The arguments for bounding $\P[\mathcal{C}_2]$ and $\P[\mathcal{C}_3]$ are almost identical; we only treat the former explicitly. Note that one step transfer operators $\begin{pmatrix} E - V_k & -1 \\ 1 & 0 \end{pmatrix}$ and $\begin{pmatrix} E' - V_k & -1 \\ 1 & 0 \end{pmatrix}$ are separated by a rank one operator of norm equal to $|E'-E|$. In particular, if $\mathcal{C}_2$ holds, then we have

\begin{align*}\exp(-\kappa L^\beta/2) &\geq \left\|A^{E'}_{[-L,-L+L^{\eta}]}\begin{pmatrix} 0 \\ 1 \end{pmatrix}\right\| \\
&\geq \left\|A^{E}_{[-L,-L+L^{\eta}]}\begin{pmatrix} 0 \\ 1 \end{pmatrix}\right\| - \left\|A^{E'}_{[-L,-L+L^{\eta}]}-A^E_{[-L,-L+L^{\eta}]}\right\|\\
&\geq \exp(-\kappa L^\beta) - \left\|A^{E'}_{[-L,-L+L^{\eta}]}-A^E_{[-L,-L+L^{\eta}]}\right\|
\end{align*}

In particular, one obtains for large $L$ the estimate $\|A^{E'}_{[-L,-L+L^{\eta}]}-A^E_{[-L,-L+L^{\eta}]}\| \geq \exp(-\kappa L^\beta)/2$. Letting $\Gamma :=\|A^{E'}_{[-L,-L+L^{\eta}]}-A^E_{[-L,-L+L^{\eta}]}\|$, by Chebyshev:
\begin{align*} \P[\mathcal{C}_2] &\leq \P[\Gamma > \exp(\kappa L^\beta)/2]\\
&\leq \P[\log^{p\star}(\Gamma) > \log^{p\star}(\exp(\kappa L^\beta)/2)]\\
&\leq C \kappa^p L^{-\beta p}\E[\log^{p\star}(\Gamma)]
\end{align*}
We can bound $\Gamma$ as follows:
\[ \Gamma \leq \sum_{k=1}^{L^{\eta}+1} \sum_{i_1 < \cdots <i_k}\exp(-k L^\beta) \left\|\begin{pmatrix} E - V_{i_k} & -1 \\ 1 & 0 \end{pmatrix}\right\| \cdots \left\|\begin{pmatrix} E- V_{i_1} & -1 \\ 1 & 0 \end{pmatrix}\right\| \]
(This is actually a dramatic overestimate; it is not necessary to count over all possible choices of $k$ indices. However, by overcounting here we can make use of nice combinatorial identities later.) Using concavity, subadditivity, the fact that it is submultiplicative up to a constant, and the fact that $\log^{p\star}$ is linear for small values, we obtain for some $C > 1$:
\begin{align*} \E[\log^{p\star}(\Gamma)] &\leq \E\left[\log^{p\star}\left(\sum_{k=1}^{L^{\eta}}\sum_{i_1 < \cdots < i_k}|E-E'|^k\|T^E_{L^{\eta}+1-k}\|\right)\right]\\
&\leq \E\left[\sum_{k=1}^{L^{\eta}}\sum_{i_1 < \cdots < i_k}\log^{p\star}(|E-E'|)^k\prod_{\ell\notin \{i_1,\cdots,i_k\}}^k \log^{p\star}(\|T^E_{\ell}\|)\right]\\
&\leq \sum_{k=1}^{L^{\eta}} \binom{L^{\eta}+1}{k}\left(\frac{p}{e}\right)^{pk}\exp(- kL^{\beta})M^{L^{\eta}+1-k}\\
&=(CM+\left(\frac{p}{e}\right)^p\exp(- L^{\beta}))^{L^{\eta}+1} - M^{L^{\eta}+1}\\
&\leq C(L^{\eta}+1)M^{L^{\eta}}\left(\frac{p}{e}\right)^p\exp(- L^{\beta})
\end{align*}
(Recall that $M = \sup_{E \in I} \E\left[\log^{p\star}\left(\|T_0^E\| \right)\right]$.) We have exploited that $\exp(- L^{\beta})$ is small and linearized $x \mapsto x^{L^{\eta}+1}$ in the last step. Because $2(L^{\eta}+1)(M+1)^{L^{\eta}}\exp(-L^{\delta\beta}))$ is small one obtains finally:

\[ \E[\log^{p\star}(\Gamma)] \leq CL^{\eta}(M+1)^{L^{\eta}}\exp(-L^{\beta})\]
for some $C > 0$.

This yields \begin{equation}\label{c2bound}\P[\mathcal{C}_2] \leq C \kappa^p L^{\eta-\beta p}(M+1)^{L^{\eta}}\exp(-L^{\beta}) \end{equation}
which is, for large $L$, dominated by $\frac{1}{2} \exp(-L^\beta)$.

A similar argument gives
\begin{equation}
    \P[\mathcal{C}_3] \leq \frac{1}{2}\exp(-L^\beta)
\end{equation}
Finally, to bound $\mathcal{C}_4$, we note that it is a union of at most $2L$ large deviation events. We make this explicit in the case of $\mathcal{A}_{\kappa,L}^E$. Clearly said event holds if and only if
\begin{equation}
    \frac{1}{L^{\eta}}\log \left\| A^E_{[-L,-L+L^{\eta}]}\begin{pmatrix} 0 \\ 1 \end{pmatrix}\right\| \geq \kappa 
\end{equation}
By \Cref{ldesthm}, if $\kappa < \lambda_{\min}$, we obtain a bound $\P[\mathcal{C}_4] \leq 2L^{1- \eta p''}$ for any $p'' \in (p', p)$. Hence we obtain
\begin{equation}
    \P[\mathcal{C}_1 \cup \mathcal{C}_2 \cup \mathcal{C}_3 \cup \mathcal{C}_4] \leq L^{1- \eta p'}
\end{equation}
for $L$ sufficiently large by combining the bounds on all terms.
\end{proof}
Finally, we prove \Cref{mainlocalthm}.
\begin{proof}[Proof of \Cref{mainlocalthm}]
If $p > 11$, one can pick $q_1 > 1$, $q_2 > 4q_1 + 6$, $\beta \in (0,1)$, $\eta \in (0,\beta)$ and $p'< p$ such that particular $1+\eta - \eta p' < -q_2$. Indeed, let $\kappa = \min\{p-11, .01\}$, and take $q_1 = 1+\kappa/16$, $q_2 = 10 + \kappa/4$, $\eta = \frac{11 + \kappa/4}{11+ \kappa/2}$, and $p' = \frac{1}{\eta}(11 + \eta + \kappa/2)$. By \Cref{initialscalethm2}, we have condition \ref{initialscale}, and by \Cref{wegthm}, we have condition \ref{Wegner}. Thus, the hypotheses of \Cref{msaassumptions} hold, and so we have almost sure localization on a small interval around any fixed energy $E$. We take a countable collection of intervals $(E_i - \delta_i, E_i + \delta_i)$ covering $\R$ for which almost sure localization holds via said theorem; taking an intersection over all the events gives almost sure localization for all energies. Such a countable collection exists by the fact that $\R$ is Lindel\"of, i.e. any open cover has a countable subcover, and the theorem follows.
\end{proof}
\bibliographystyle{amsalpha}
\bibliography{uniflyap}

\newcommand{\etalchar}[1]{$^{#1}$}
\providecommand{\bysame}{\leavevmode\hbox to3em{\hrulefill}\thinspace}
\providecommand{\MR}{\relax\ifhmode\unskip\space\fi MR }
\providecommand{\MRhref}[2]{%
  \href{http://www.ams.org/mathscinet-getitem?mr=#1}{#2}
}
\providecommand{\href}[2]{#2}
\begin{thebibliography}{BDF{\etalchar{+}}17}

\bibitem[AM93]{aizenman1993localization}
Michael Aizenman and Stanislav Molchanov, \emph{{Localization at large disorder and at extreme energies: An elementary derivations}}, Communications in Mathematical Physics \textbf{157} (1993), 245--278.

\bibitem[AW15]{aizenman2015random}
Michael Aizenman and Simone Warzel, \emph{Random operators}, vol. 168, American Mathematical Soc., 2015.

\bibitem[BDF{\etalchar{+}}17]{Bucaj2017LocalizationFT}
Valmir Bucaj, David Damanik, Jake Fillman, Vitaly Gerbuz, Tom VandenBoom, Fengpeng Wang, and Zhenghe Zhang, \emph{{Localization for the one-dimensional Anderson model via positivity and large deviations for the Lyapunov exponent}}, Transactions of the American Mathematical Society (2017).

\bibitem[BK05]{bourgain2005localization}
Jean Bourgain and Carlos~E Kenig, \emph{{On localization in the continuous Anderson-Bernoulli model in higher dimension}}, Inventiones mathematicae \textbf{161} (2005), no.~2.

\bibitem[CFKS09]{cycon2009schrodinger}
Hans~L Cycon, Richard~G Froese, Werner Kirsch, and Barry Simon, \emph{Schr{\"o}dinger operators: With application to quantum mechanics and global geometry}, Springer, 2009.

\bibitem[CKM87]{carmona1987anderson}
Ren{\'e} Carmona, Abel Klein, and Fabio Martinelli, \emph{{Anderson localization for Bernoulli and other singular potentials}}, Communications in Mathematical Physics \textbf{108} (1987), no.~1, 41--66.

\bibitem[DF25]{damanik2025one}
David Damanik and Jake Fillman, \emph{One-dimensional ergodic schr{\"o}dinger operators: Ii. specific classes}, vol. 249, American Mathematical Society, 2025.

\bibitem[DGK25]{damanik2025localization}
David Damanik, Anton Gorodetski, and Victor Kleptsyn, \emph{Localization for random schr$\backslash$" odinger operators defined by block factors}, arXiv preprint arXiv:2504.08153 (2025).

\bibitem[DK16]{duarte2016lyapunov}
Pedro Duarte and Silvius Klein, \emph{Lyapunov exponents of linear cocycles}, Atlantis Studies in Dynamical Systems \textbf{3} (2016).

\bibitem[DK20]{duarte2020large}
\bysame, \emph{Large deviations for products of random two dimensional matrices}, Communications in Mathematical Physics \textbf{375} (2020), 2191--2257.

\bibitem[DS20]{ding2020localization}
Jian Ding and Charles~K Smart, \emph{{Localization near the edge for the Anderson Bernoulli model on the two dimensional lattice}}, Inventiones mathematicae \textbf{219} (2020), 467--506.

\bibitem[FK60]{FurstenbergKesten}
H.~Furstenberg and H.~Kesten, \emph{{Products of Random Matrices}}, The Annals of Mathematical Statistics \textbf{31} (1960), no.~2, 457 -- 469.

\bibitem[FK83]{FurstenbergKifer}
H.~Furstenberg and Y.~Kifer, \emph{Random matrix products and measures on projective spaces}, Israel Journal of Mathematics \textbf{46} (1983), no.~1, 12--32.

\bibitem[FMSS85]{frohlich1985constructive}
J{\"u}rg Fr{\"o}hlich, Fabio Martinelli, Elisabetta Scoppola, and Thomas Spencer, \emph{{Constructive proof of localization in the Anderson tight binding model}}, Communications in Mathematical Physics \textbf{101} (1985), no.~1, 21--46.

\bibitem[FS83]{frohlich1983absence}
J{\"u}rg Fr{\"o}hlich and Thomas Spencer, \emph{{Absence of diffusion in the Anderson tight binding model for large disorder or low energy}}, Communications in Mathematical Physics \textbf{88} (1983), no.~2, 151--184.

\bibitem[Fur63]{FurstenbergNRP}
Harry Furstenberg, \emph{Noncommuting random products}, Trans. Amer. Math. Soc. \textbf{108} (1963), 377--428. \MR{163345}

\bibitem[GK01]{germinet2001bootstrap}
Fran{\c{c}}ois Germinet and Abel Klein, \emph{Bootstrap multiscale analysis and localization in random media}, Communications in Mathematical Physics \textbf{222} (2001), no.~2, 415--448.

\bibitem[GK12]{klein2012comprehensive}
\bysame, \emph{{A comprehensive proof of localization for continuous Anderson models with singular random potentials}}, Journal of the European Mathematical Society \textbf{15} (2012), no.~1, 53--143.

\bibitem[GK21]{gorodetski2021parametric}
Anton Gorodetski and Victor Kleptsyn, \emph{{Parametric Furstenberg theorem on random products of $\mathrm{SL}(2, \mathbb{R})$ matrices}}, Advances in Mathematics \textbf{378} (2021), 107522.

\bibitem[GK25]{gorodetski2025non}
\bysame, \emph{Non-stationary anderson localization}, Communications of the American Mathematical Society \textbf{5} (2025), no.~03, 81--143.

\bibitem[GMP77]{goldsheid1977random}
Ilya Goldsheid, Stanislav Molchanov, and Leonid Pastur, \emph{{A random homogeneous Schr{\"o}dinger operator has a pure point spectrum}}, Funct. Anal. Appl \textbf{11} (1977), no.~1, 1--10.

\bibitem[HR25]{hurtado2025uniformestimatesrandommatrix}
Omar Hurtado and Sidhanth Raman, \emph{Uniform estimates for random matrix products and applications}, arXiv:2504:08226 (2025).

\bibitem[Hur23]{hurtado2023lifting}
Omar Hurtado, \emph{{A “lifting” method for exponential large deviation estimates and an application to certain non-stationary 1D lattice Anderson models}}, Journal of Mathematical Physics \textbf{64} (2023), no.~6.

\bibitem[JZ19]{Jitomirskaya2019}
Svetlana Jitomirskaya and Xiaowen Zhu, \emph{{Large Deviations of the Lyapunov Exponent and Localization for the 1D Anderson Model}}, Communications in Mathematical Physics \textbf{370} (2019), no.~1, 311--324.

\bibitem[KS80]{kunz1980spectre}
Herv{\'e} Kunz and Bernard Souillard, \emph{Sur le spectre des op{\'e}rateurs aux diff{\'e}rences finies al{\'e}atoires}, Communications in Mathematical Physics \textbf{78} (1980), 201--246.

\bibitem[LP82]{les1982theoremes}
DE~MATRICES~ALEATOIRES LES~PRODUITS, \emph{Theoremes limites pour les produits de matrices aleatoires emile le page laboratoire de probabilit{\'e} universit{\'e} de rennes}, Probability Measures on Groups: Proceedings of the Sixth Conference Held at Oberwolfach, Germany, June 28-July 4, 1981, vol. 928, Springer, 1982, p.~258.

\bibitem[LZ22]{li2022anderson}
Linjun Li and Lingfu Zhang, \emph{{Anderson--Bernoulli localization on the three-dimensional lattice and discrete unique continuation principle}}, Duke Mathematical Journal \textbf{171} (2022), no.~2, 327--415.

\bibitem[MS22]{Macera2022}
Davide Macera and Sasha Sodin, \emph{{Anderson Localisation for Quasi-One-Dimensional Random Operators}}, Annales Henri Poincar{\'e} \textbf{23} (2022), no.~12, 4227--4247.

\bibitem[Ran19]{rangamani2019singular}
Nishant Rangamani, \emph{{Singular-unbounded random Jacobi matrices}}, Journal of Mathematical Physics \textbf{60} (2019), no.~8.

\bibitem[RZ23]{rangamani2023dynamical}
Nishant Rangamani and Xiaowen Zhu, \emph{Dynamical localization for the singular anderson model in $\mathbf{Z}^d$}, arXiv preprint arXiv:2307.01608 (2023).

\bibitem[SVW98]{shubin1998some}
C~Shubin, R~Vakilian, and T~Wolff, \emph{Some harmonic analysis questions suggested by anderson-bernoulli models}, Geometric \& Functional Analysis GAFA \textbf{8} (1998), no.~5, 932--964.

\bibitem[Tsa99]{tsay1999some}
Jhishen Tsay, \emph{Some uniform estimates in products of random matrices}, Taiwanese Journal of Mathematics \textbf{3} (1999), no.~3, 291--302.

\bibitem[vDK89]{von1989new}
Henrique von Dreifus and Abel Klein, \emph{{A new proof of localization in the Anderson tight binding model}}, Communications in Mathematical Physics \textbf{124} (1989), 285--299.

\end{thebibliography}

\end{document}